\documentclass[12pt]{amsart}%
\usepackage{amsmath}
\usepackage{graphicx}
\usepackage{amsfonts}
\usepackage{amssymb}%
\newtheorem{theorem}{Theorem}[section]
\newtheorem{lemma}[theorem]{Lemma}

\newtheorem{proposition}[theorem]{Proposition}

\theoremstyle{definition}
\newtheorem{definition}[theorem]{Definition}
\newtheorem{example}[theorem]{Example}
\theoremstyle{remark}

\numberwithin{equation}{section}

\begin{document}
\title[Thick Distributional Calculus]{Applications of the Thick Distributional Calculus}

\begin{abstract}
We give several applications of the thick distributional calculus. We consider
homogeneous thick distributions, point source fields, and higher order
derivatives of order $0.$

\end{abstract}
\author{Yunyun Yang}
\address{Department of Mathematics\\
Louisiana State University\\
Baton Rouge, LA 70803}
\email{yyang18@math.lsu.edu }
\author{Ricardo Estrada}
\email{restrada@math.lsu.edu}
\thanks{The authors gratefully acknowledge support from NSF, through grant number 0968448.}
\maketitle

\section{Introduction\label{Section: Intro}}

The aim of this note is to give several applications of the recently
introduced calculus of thick distributions in several variables \cite{YE2},
generalizing the thick distributions of one variable \cite{EF07}. The thick
distributional calculus allows us to study problems where a finite number of
special points are present; it is the distributional version of the analysis
of Blanchet and Faye \cite{BlanchetFaye}, who employed the concepts of
Hadamard finite parts as developed by Sellier \cite{Sellier3} to study
dynamics of point particles in high post-Newtonian approximations of general
relativity. We give a short summary of the theory of thick distributions in
Section \ref{Section: Thick distributions}.

Our first application, given in Section \ref{Section; The distribution 1}, is
the computation of the distributional derivatives of homogeneous distributions
in $\mathbb{R}^{n}$ by first computing the thick distributional derivatives
and then projecting onto the space of standard distributions. Our analysis
makes several delicate points quite clear.

Next, in Section \ref{Section; Bowen}, we consider an application to point
source fields. In \cite{Bowen}, Bowen computed the derivative of the
distribution
\begin{equation}
g_{j_{1},\ldots,j_{k}}\left(  \mathbf{x}\right)  =\frac{n_{j_{1}}\cdots
n_{j_{k}}}{r^{2}}\,, \label{IN.1}%
\end{equation}
of $\mathcal{D}^{\prime}\left(  \mathbb{R}^{3}\right)  ,$ where $r=\left\vert
\mathbf{x}\right\vert $ and $\mathbf{n=}\left(  n_{i}\right)  $ is the unit
normal vector to a sphere centered at the origin, that is, $n_{i}=x_{i}/r.$
His result can be written as\footnote{Following the notation introduced by the
late Professor Farassat \cite{Farassat}, we shall denote distributional
derivatives with an overbar.}%
\begin{equation}
\frac{\overline{\partial}}{\partial x_{i}}g_{j_{1},\ldots,j_{k}}=\left\{
\sum_{q=1}^{k}\delta_{ij_{q}}\frac{n_{j_{1}}\cdots n_{j_{k}}}{n_{j_{q}}%
}-\left(  k+2\right)  n_{i}n_{j_{1}}\cdots n_{j_{k}}\right\}  \frac{1}{r^{3}%
}+A\delta\left(  \mathbf{x}\right)  \,, \label{IN.2}%
\end{equation}
where $n_{i}n_{j_{1}}\cdots n_{j_{k}}=n_{1}^{a}n_{2}^{b}n_{3}^{c},$ and $A=0$
if $a,$ $b,$ or $c$ is odd, while%
\begin{equation}
A=\frac{2\Gamma\left(  \left(  a+1\right)  /2\right)  \Gamma\left(  \left(
b+1\right)  /2\right)  \Gamma\left(  \left(  c+1\right)  /2\right)  }%
{\Gamma\left(  \left(  a+b+c+3\right)  /2\right)  }\,, \label{IN.3}%
\end{equation}
if the three exponents are even. Interestingly, he observes that if one tries
to compute this formula by induction, employing the product rule for
derivatives, the result obtained is \emph{wrong.} In this article we show that
one can actually apply the product rule in the space of thick distributions,
obtaining (\ref{IN.2}) by induction; furthermore, our analysis shows
\emph{why} the wrong result is obtained when applying the product rule in
\cite{Bowen}.

Finally in Section \ref{Section: Higher order derivatives} we show how the
thick distributional calculus allows one to avoid mistakes in the computation
of higher order derivatives of thick distributions of order $0.$

\section{Thick distributions\label{Section: Thick distributions}}

We now recall the basic ideas of the thick distributional calculus \cite{YE2}.
If $\mathbf{a}$ is a fixed point of $\mathbb{R}^{n},$ then the space of test
functions with a thick point at $\mathbf{x=a}$ is defined as
follows.\smallskip

\begin{definition}
\label{Definition Test}Let $\mathcal{D}_{\ast,\mathbf{a}}\left(
\mathbb{R}^{n}\right)  $ denote the vector space of all smooth functions
$\phi$ defined in $\mathbb{R}^{n}\setminus\left\{  \mathbf{a}\right\}  ,$ with
support of the form $K\setminus\left\{  \mathbf{a}\right\}  ,$ where $K$ is
compact in $\mathbb{R}^{n},$ that admit a strong asymptotic expansion of the
form
\begin{equation}
\phi\left(  \mathbf{a}+\mathbf{x}\right)  =\phi\left(  \mathbf{a}%
+r\mathbf{w}\right)  \sim\underset{j=m}{\overset{\infty}{\sum}}a_{j}\left(
\mathbf{w}\right)  r^{j},\ \ \ \text{as }\mathbf{x}\rightarrow\mathbf{0}\,,
\label{Def.1}%
\end{equation}
where $m$ is an integer (positive or negative), and where the $a_{j}$ are
smooth functions of $\mathbf{w},$ that is, $a_{j}\in\mathcal{D}\left(
\mathbb{S}\right)  .$ The subspace $\mathcal{D}_{\ast,\mathbf{a}}^{\left[
m\right]  }\left(  \mathbb{R}^{n}\right)  $ consists of those test functions
$\phi$ whose expansion (\ref{Def.1}) begins at $m.$ For a fixed compact $K$
whose interior contains $\mathbf{a},$ $\mathcal{D}_{\ast,\mathbf{a}}^{\left[
m;K\right]  }\left(  \mathbb{R}^{n}\right)  $ is the subspace formed by those
test functions of $\mathcal{D}_{\ast,\mathbf{a}}^{\left[  m\right]  }\left(
\mathbb{R}^{n}\right)  $ that vanish in $\mathbb{R}^{n}\setminus K.$\smallskip
\end{definition}

Observe that we require the asymptotic development of $\phi\left(
\mathbf{x}\right)  $ as $\mathbf{x}\rightarrow\mathbf{a}$ to be
\textquotedblleft strong\textquotedblright. This means \cite[Chapter
1]{GreenBook}\ that for any differentiation operator $\left(  \partial
/\partial\mathbf{x}\right)  ^{\mathbf{p}}=\left(  \partial^{p_{1}}%
...\partial^{p_{n}}\right)  /\partial x_{1}^{p_{1}}...\partial x_{n}^{p_{n}},$
the asymptotic development of $\left(  \partial/\partial\mathbf{x}\right)
^{\mathbf{p}}\phi\left(  \mathbf{x}\right)  $ as $\mathbf{x}\rightarrow
\mathbf{a}$ exists and is equal to the term-by-term differentiation of
$\sum_{j=m}^{\infty}a_{j}\left(  \mathbf{w}\right)  r^{j}.$ Observe that
saying that the expansion exists as $\mathbf{x}\rightarrow\mathbf{0}$ is the
same as saying that it exits as $r\rightarrow0,$ uniformly with respect to
$\mathbf{w}.$

We call $\mathcal{D}_{\ast,\mathbf{a}}\left(  \mathbb{R}^{n}\right)  $ the
space of test functions on $\mathbb{R}^{n}$ with a thick point located
at\textbf{\ }$\mathbf{x}=\mathbf{a}.$ We denote $\mathcal{D}_{\ast,\mathbf{0}%
}\left(  \mathbb{R}^{n}\right)  $ as $\mathcal{D}_{\ast}\left(  \mathbb{R}%
^{n}\right)  .$

The topology of the space of thick test functions is constructed as
follows.\smallskip

\begin{definition}
\label{DefinitionTopology}Let $m$ be a fixed integer and $K$ a compact subset
of $\mathbb{R}^{n}$ whose interior contains $\mathbf{a}.$ The topology of
$\mathcal{D}_{\ast,\mathbf{a}}^{\left[  m;K\right]  }\left(  \mathbb{R}%
^{n}\right)  $ is given by the seminorms $\left\{  \left\Vert
\ \ \ \ \right\Vert _{q,s}\right\}  _{q>m,s\geq0}$ defined as
\begin{equation}
\left\vert \left\vert \phi\right\vert \right\vert _{q,s}=\sup_{\mathbf{x-a}\in
K}\sup_{\left\vert \mathbf{p}\right\vert \leq s}\frac{\left\vert
\displaystyle\frac{\partial^{\mathbf{p}}\phi}{\partial\mathbf{x}}\left(
\mathbf{a+x}\right)  -\underset{j=m-\left\vert \mathbf{p}\right\vert
}{\overset{q-1}{\sum}}a_{j,\mathbf{p}}\left(  \mathbf{w}\right)
r^{j}\right\vert }{r^{q}}\,, \label{def 2}%
\end{equation}
where $\mathbf{x}=r\mathbf{w}$ and%
\begin{equation}
\displaystyle\frac{\partial^{\mathbf{p}}\phi}{\partial\mathbf{x}}\left(
\mathbf{a+x}\right)  \sim\underset{j=m-\left\vert \mathbf{p}\right\vert
}{\overset{\infty}{\sum}}a_{j,\mathbf{p}}\left(  \mathbf{w}\right)  r^{j}.
\label{def 2b}%
\end{equation}

The topology of $\mathcal{D}_{\ast,\mathbf{a}}^{\left[  m\right]  }\left(
\mathbb{R}^{n}\right)  $ is the inductive limit topology of the $\mathcal{D}%
_{\ast,\mathbf{a}}^{\left[  m;K\right]  }\left(  \mathbb{R}^{n}\right)  $ as
$K\nearrow\infty.$ The topology of $\mathcal{D}_{\ast,\mathbf{a}}\left(
\mathbb{R}^{n}\right)  $ is the inductive limit topology of the $\mathcal{D}%
_{\ast,\mathbf{a}}^{\left[  m\right]  }\left(  \mathbb{R}^{n}\right)  $ as
$m\searrow-\infty.$\smallskip
\end{definition}

A sequence $\left\{  \phi_{l}\right\}  _{l=0}^{\infty}$ in $\mathcal{D}%
_{\ast,\mathbf{a}}\left(  \mathbb{R}^{n}\right)  $ converges to $\psi$ if and
only there exists $l_{0}\geq0,$ an integer $m,$ and a compact set $K$ with
$\mathbf{a}$ in its interior, such that $\phi_{l}\in\mathcal{D}_{\ast
,\mathbf{a}}^{\left[  m;K\right]  }\left(  \mathbb{R}^{n}\right)  $ for $l\geq
l_{0}$ and $\left\vert \left\vert \psi-\phi_{l}\right\vert \right\vert
_{q,s}\rightarrow0$ as $l\rightarrow\infty$ if $q>m,s\geq0.$ Notice that if
$\left\{  \phi_{l}\right\}  _{l=0}^{\infty}$ converges to $\psi$ in
$\mathcal{D}_{\ast,\mathbf{a}}\left(  \mathbb{R}^{n}\right)  $ then $\phi_{l}$
and the corresponding derivatives converge uniformly to $\psi$ and its
derivatives in any set of the form $\mathbb{R}^{n}\setminus B,$ where $B$ is a
ball with center at $\mathbf{a};$ in fact, $r^{\left\vert \mathbf{p}%
\right\vert -m}\left(  \partial/\partial\mathbf{x}\right)  ^{\mathbf{p}}%
\phi_{l}$ converges uniformly to $r^{\left\vert \mathbf{p}\right\vert
-m}\left(  \partial/\partial\mathbf{x}\right)  ^{\mathbf{p}}\psi$ over all
$\mathbb{R}^{n}.$ Furthermore, if $\left\{  a_{j}^{l}\right\}  $ are the
coefficients of the expansion of $\phi_{l}$ and $\left\{  b_{j}\right\}  $ are
those for $\psi,$ then $a_{j}^{l}\rightarrow b_{j}$ in the space
$\mathcal{D}\left(  \mathbb{S}\right)  $ for each $j\geq m.$

We can now consider distributions in a space with one thick point, the
\textquotedblleft thick distributions.\textquotedblright\smallskip

\begin{definition}
The space of distributions on $\mathbb{R}^{n}$ with a thick point at
$\mathbf{x=a}$ is the dual space of $\mathcal{D}_{\ast,\mathbf{a}}\left(
\mathbb{R}^{n}\right)  .$ We denote it $\mathcal{D}_{\ast,\mathbf{a}%
}^{^{\prime}}\left(  \mathbb{R}^{n}\right)  ,$ or just as $\mathcal{D}_{\ast
}^{^{\prime}}\left(  \mathbb{R}^{n}\right)  $ when $\mathbf{a=0}.$\smallskip
\end{definition}

Observe that $\mathcal{D}\left(  \mathbb{R}^{n}\right)  ,$\ the space of
standard test functions, is a closed subspace of $\mathcal{D}_{\ast
,\mathbf{a}}\left(  \mathbb{R}^{n}\right)  ;$ we denote by%
\begin{equation}
i:\mathcal{D}\left(  \mathbb{R}^{n}\right)  \rightarrow\mathcal{D}%
_{\ast,\mathbf{a}}\left(  \mathbb{R}^{n}\right)  \,, \label{Test.3}%
\end{equation}
the inclusion map and by%
\begin{equation}
\Pi:\mathcal{D}_{\ast,\mathbf{a}}^{\prime}\left(  \mathbb{R}^{n}\right)
\rightarrow\mathcal{D}^{\prime}\left(  \mathbb{R}^{n}\right)  \,,
\label{Dis.1}%
\end{equation}
the projection operator, dual of the inclusion (\ref{Test.3}).

The derivatives of thick distributions are defined in much the same way as the
usual distributional derivatives, that is, by duality.\smallskip

\begin{definition}
\label{Definition: Derivaative}If $f\in\mathcal{D}_{\ast,\mathbf{a}}^{\prime
}\left(  \mathbb{R}^{n}\right)  $ then its thick distributional derivative
$\partial^{\ast}f/\partial x_{j}$ is defined as%
\begin{equation}
\left\langle \frac{\partial^{\ast}f}{\partial x_{j}},\phi\right\rangle
=-\left\langle f,\frac{\partial\phi}{\partial x_{j}}\right\rangle
\,,\ \ \ \ \ \phi\in\mathcal{D}_{\ast,\mathbf{a}}\left(  \mathbb{R}%
^{n}\right)  \,. \label{Der.1}%
\end{equation}

\end{definition}

We denote by $\mathcal{E}_{\ast}\left(  \mathbb{R}^{n}\right)  $ the space of
smooth functions in $\mathbb{R}^{n}\setminus\left\{  \mathbf{a}\right\}  $
that have a strong asymptotic expansion of the form (\ref{Def.1});
alternatively, $\psi\in\mathcal{E}_{\ast}\left(  \mathbb{R}^{n}\right)  $\ if
$\psi=\psi_{1}+\psi_{2},$ where $\psi_{1}\in\mathcal{E}\left(  \mathbb{R}%
^{n}\right)  $\footnote{In general $\mathcal{E}\left(  U\right)  $ is the
space of all smooth functions in the open set $U.$} and where $\psi_{2}%
\in\mathcal{D}_{\ast}\left(  \mathbb{R}^{n}\right)  .$ The space
$\mathcal{E}_{\ast}\left(  \mathbb{R}^{n}\right)  $ is the space of
\emph{multipliers} of $\mathcal{D}_{\ast}\left(  \mathbb{R}^{n}\right)  $ and
of $\mathcal{D}_{\ast}^{\prime}\left(  \mathbb{R}^{n}\right)  .$ Furthermore
\cite{YE2}, the product rule for derivatives holds,%
\begin{equation}
\frac{\partial^{\ast}\left(  \psi f\right)  }{\partial x_{j}}=\frac
{\partial\psi}{\partial x_{j}}f+\psi\frac{\partial^{\ast}f}{\partial x_{j}}\,,
\label{Der.3}%
\end{equation}
if $f$ is a thick distribution and $\psi$ is a multiplier. Notice that
$\partial\psi/\partial x_{j}$ is the \emph{ordinary} derivative in
(\ref{Der.3}).

Let $g\left(  \mathbf{w}\right)  $ is a distribution in $\mathbb{S}.$ The
thick delta function of degree $q,$ denoted as $g\delta_{\ast}^{\left[
q\right]  },$ or as $g\left(  \mathbf{w}\right)  \delta_{\ast}^{\left[
q\right]  },$ acts on a thick test function $\phi\left(  \mathbf{x}\right)  $
as
\begin{equation}
\left\langle g\delta_{\ast}^{\left[  q\right]  },\phi\right\rangle
_{\mathcal{D}_{\ast}^{\prime}\left(  \mathbb{R}^{n}\right)  \times
\mathcal{D}_{\ast}\left(  \mathbb{R}^{n}\right)  }=\frac{1}{C}\left\langle
g\left(  \mathbf{w}\right)  ,a_{q}\left(  \mathbf{w}\right)  \right\rangle
_{\mathcal{D}^{\prime}\left(  \mathbb{S}\right)  \times\mathcal{D}\left(
\mathbb{S}\right)  }\,, \label{Dis.15}%
\end{equation}
where $\phi\left(  r\mathbf{w}\right)  \sim\sum_{j=m}^{\infty}a_{j}\left(
\mathbf{w}\right)  r^{j},$ as $r\rightarrow0^{+},$ and where
\begin{equation}
C=\frac{2\pi^{n/2}}{\Gamma\left(  n/2\right)  }\,, \label{Area}%
\end{equation}
is the surface area of the unit sphere $\mathbb{S}$ of $\mathbb{R}^{n}.$ If
$g$ is locally integrable function in $\mathbb{S},$ then
\begin{equation}
\left\langle g\delta_{\ast}^{\left[  q\right]  },\phi\right\rangle
_{\mathcal{D}_{\ast}^{\prime}\left(  \mathbb{R}^{n}\right)  \times
\mathcal{D}_{\ast}\left(  \mathbb{R}^{n}\right)  }=\frac{1}{C}\int
_{\mathbb{S}}g\left(  \mathbf{w}\right)  a_{q}\left(  \mathbf{w}\right)
\,\mathrm{d}\sigma\left(  \mathbf{w}\right)  \,. \label{Dis.16}%
\end{equation}
Thick deltas of order $0$ are called just thick deltas, and we shall use the
notation $g\delta_{\ast}$ instead of $g\delta_{\ast}^{\left[  0\right]  }.$

Let $g\in\mathcal{D}^{\prime}\left(  \mathbb{S}\right)  .$ Then%
\begin{equation}
\frac{\partial^{\ast}}{\partial x_{j}}\left(  g\delta_{\ast}^{\left[
q\right]  }\right)  =\left(  \frac{\delta g}{\delta x_{j}}-\left(  q+n\right)
n_{j}g\right)  \delta_{\ast}^{\left[  q+1\right]  }. \label{DD.1}%
\end{equation}
Here $\delta g/\delta x_{j}$ is the $\delta-$derivative of $g$ \cite{EK1,
EK2}; in general the $\delta-$derivatives can be applied to functions and
distributions defined only on a smooth hypersurface $\Sigma$ of $\mathbb{R}%
^{n}.$ Suppose now that the surface is $\mathbb{S},$ the unit sphere in
$\mathbb{R}^{n}$ and let $f$ be a smooth function defined in $\mathbb{S},$
that is, $f\left(  \mathbf{w}\right)  $ is defined if $\mathbf{w}\in
\mathbb{R}^{n}$ satisfies $\left\vert \mathbf{w}\right\vert =1.$ Observe that
the expressions $\partial f/\partial x_{j}$ are not defined and, likewise, if
$\mathbf{w}=\left(  w_{j}\right)  _{1\leq j\leq n}$ the expressions $\partial
f/\partial w_{j}$ do not make sense either; the derivatives that are always
defined and that one should consider are the $\delta f/\delta x_{j},$ $1\leq
j\leq n.$ Let $F_{0}$\ be the extension of $f$\ to $\mathbb{R}^{n}%
\setminus\left\{  \mathbf{0}\right\}  $ that is homogeneous of degree $0,$
namely, $F_{0}\left(  \mathbf{x}\right)  =f\left(  \mathbf{x}/r\right)  $
where $r=\left\vert \mathbf{x}\right\vert .$ Then \cite{YE2}%
\begin{equation}
\frac{\delta f}{\delta x_{j}}=\left.  \frac{\partial F_{0}}{\partial x_{j}%
}\right\vert _{\mathbb{S}}\,. \label{Su.4}%
\end{equation}
Also, if we use polar coordinates, $\mathbf{x}=r\mathbf{w},$ so that
$F_{0}\left(  \mathbf{x}\right)  =f\left(  \mathbf{w}\right)  ,$ then
$\partial F_{0}/\partial x_{j}$ is homogeneous of degree $-1,$ and actually
$\partial F_{0}/\partial x_{j}=r^{-1}\delta f/\delta x_{j}$ if $\mathbf{x\neq
0.}$

The matrix $\mathbf{\mu=}\left(  \mu_{ij}\right)  _{1\leq i,j\leq n},$ where
$\mu_{ij}=\delta n_{i}/\delta x_{j},$ plays an important role in the study of
distributions on a surface $\Sigma.$ If $\Sigma=\mathbb{S}$ then $\mu
_{ij}=\delta n_{i}/\delta x_{j}=\delta_{ij}-n_{i}n_{j}.$ Observe that
$\mu_{ij}=\mu_{ji},$ an identity that holds in any surface.

The differential operators $\delta f/\delta x_{j}$\ are initially defined if
$f$ is a smooth function defined on $\Sigma,$ but we can also define them when
$f$ is a distribution. We can do this if we use the fact that smooth functions
are dense in the space of distributions on $\Sigma.$

\section{The thick distribution $\mathcal{P}f\left(  1\right)  $%
\label{Section; The distribution 1}}

Let us consider one of the simplest functions, namely, the function $1,$
defined in $\mathbb{R}^{n}.$ Naturally this function is locally integrable,
and thus it defines a regular distribution, also denoted as $1,$ and the
ordinary derivatives and the distributional derivatives both coincide and give
the value $0.$ On the other hand, $1$ does not automatically give an element
of $\mathcal{D}_{\ast}^{\prime}\left(  \mathbb{R}^{n}\right)  $ since if
$\phi\in\mathcal{D}_{\ast}\left(  \mathbb{R}^{n}\right)  $ the integral
$\int_{\mathbb{R}^{n}}\phi\left(  \mathbf{x}\right)  \,\mathrm{d}\mathbf{x}$
could be divergent, and thus we consider the \emph{spherical} finite
part\footnote{If instead of removing balls of radius $\varepsilon,$ solids of
other shapes are removed one obtains a different thick distribution.} thick
distribution $\mathcal{P}f\left(  1\right)  $ given as%
\begin{equation}
\left\langle \mathcal{P}f\left(  1\right)  ,\phi\right\rangle =\mathrm{F.p.}%
\int_{\mathbb{R}^{n}}\phi\left(  \mathbf{x}\right)  \,\mathrm{d}%
\mathbf{x=}\mathrm{F.p.}\lim_{\varepsilon\rightarrow0^{+}}\int_{\left\vert
\mathbf{x}\right\vert \geq\varepsilon}\phi\left(  \mathbf{x}\right)
\,\mathrm{d}\mathbf{x}\,. \label{T1.1}%
\end{equation}
The derivatives of $\mathcal{P}f\left(  1\right)  $ do not vanish, since
actually we have the following formula \cite{YE2}.\smallskip

\begin{lemma}
\label{Lemma T1.1}In $\mathcal{D}_{\ast}^{\prime}\left(  \mathbb{R}%
^{n}\right)  ,$
\begin{equation}
\frac{\partial^{\ast}}{\partial x_{i}}\left(  \mathcal{P}f\left(  1\right)
\right)  =Cn_{i}\delta_{\ast}^{\left[  -n+1\right]  }, \label{T1.1p}%
\end{equation}
where $C$ is given by (\ref{Area}).
\end{lemma}

\begin{proof}
One can find a proof of a more general statement in \cite{YE2}, but in this
simpler case the proof can be written as follows,%
\begin{align*}
\left\langle \frac{\partial^{\ast}}{\partial x_{i}}\left(  \mathcal{P}f\left(
1\right)  \right)  ,\phi\right\rangle  &  =-\left\langle \mathcal{P}f\left(
1\right)  ,\frac{\partial\phi}{\partial x_{i}}\right\rangle \\
&  =-\mathrm{F.p.}\lim_{\varepsilon\rightarrow0^{+}}\int_{\left\vert
\mathbf{x}\right\vert \geq\varepsilon}\frac{\partial\phi}{\partial x_{i}%
}\,\mathrm{d}\mathbf{x}\\
&  =\mathrm{F.p.}\lim_{\varepsilon\rightarrow0^{+}}\int_{\varepsilon
\mathbb{S}^{n-1}}n_{i}\phi\,\mathrm{d}\sigma\,,
\end{align*}
so that if $\phi\in\mathcal{D}_{\ast}\left(  \mathbb{R}^{n}\right)  $ has the
expansion $\phi\left(  \mathbf{x}\right)  \sim\sum_{j=m}^{\infty}a_{j}\left(
\mathbf{w}\right)  r^{j},$ as $\mathbf{x}\rightarrow\mathbf{0},$ then%
\[
\int_{\varepsilon\mathbb{S}^{n-1}}n_{i}\phi\,\mathrm{d}\sigma\sim\sum
_{j=m}^{\infty}\left(  \int_{\mathbb{S}}n_{i}a_{j}\left(  \mathbf{w}\right)
\,\mathrm{d}\sigma\left(  \mathbf{w}\right)  \right)  \varepsilon^{n-1+j},
\]
as $\varepsilon\rightarrow0^{+}.$ The finite part of the limit is equal to the
coefficient of $\varepsilon^{0},$ thus%
\begin{align*}
\mathrm{F.p.}\underset{\varepsilon\rightarrow0}{\lim}\int_{\varepsilon
\mathbb{S}^{n-1}}n_{i}\phi\,\mathrm{d}\sigma &  =\int_{\mathbb{S}}n_{i}%
a_{1-n}\left(  \mathbf{w}\right)  \,\mathrm{d}\sigma\left(  \mathbf{w}\right)
\\
&  =\left\langle Cn_{i}\delta_{\ast}^{\left[  1-n\right]  },\phi\right\rangle
\,,
\end{align*}
as required.\smallskip
\end{proof}

If $\psi\in\mathcal{E}_{\ast}\left(  \mathbb{R}^{n}\right)  $ is a multiplier
of $\mathcal{D}_{\ast}\left(  \mathbb{R}^{n}\right)  ,$ then we define, in a
similar way, the thick distribution $\mathcal{P}f\left(  \psi\right)
\in\mathcal{D}_{\ast}^{\prime}\left(  \mathbb{R}^{n}\right)  ,$ and we clearly
have the useful formula%
\begin{equation}
\mathcal{P}f\left(  \psi\right)  =\psi\mathcal{P}f\left(  1\right)  \,,
\label{T1.2}%
\end{equation}
which immediately gives the thick distributional derivative of $\mathcal{P}%
f\left(  \psi\right)  $ as%
\[
\frac{\partial^{\ast}}{\partial x_{i}}\left(  \mathcal{P}f\left(  \psi\right)
\right)  =\frac{\partial\psi}{\partial x_{i}}\mathcal{P}f\left(  1\right)
+\psi\frac{\partial^{\ast}}{\partial x_{i}}\left(  \mathcal{P}f\left(
1\right)  \right)  \,,
\]
so that we obtain the ensuing formula.\smallskip

\begin{proposition}
\label{Prop. T1.1}If $\psi\in\mathcal{E}_{\ast}\left(  \mathbb{R}^{n}\right)
$ then%
\begin{equation}
\frac{\partial^{\ast}}{\partial x_{i}}\left(  \mathcal{P}f\left(  \psi\right)
\right)  =\mathcal{P}f\left(  \frac{\partial\psi}{\partial x_{i}}\right)
+Cn_{i}\psi\delta_{\ast}^{\left[  1-n\right]  }. \label{T1.3}%
\end{equation}

\end{proposition}

Notice that, in general, the term $Cn_{i}\psi\delta_{\ast}^{\left[
1-n\right]  }$ is \emph{not} a thick delta of order $1-n.$ Indeed, let us now
consider the case when $\psi\in\mathcal{E}_{\ast}\left(  \mathbb{R}%
^{n}\right)  $ is homogeneous of order $k\in\mathbb{Z}.$ Then $\psi\left(
\mathbf{x}\right)  =r^{k}\psi_{0}\left(  \mathbf{x}\right)  ,$ where $\psi
_{0}$ is homogeneous of order $0.$ Since $r^{k}\delta_{\ast}^{\left[
q\right]  }=\delta_{\ast}^{\left[  q-k\right]  }$\cite[Eqn. (5.16)]{YE2} we
obtain the following particular case of (\ref{T1.3}), where now the term
$Cn_{i}\psi_{0}\delta_{\ast}^{\left[  1-n-k\right]  }$ \emph{is} a thick delta
of order $1-n-k.$\smallskip

\begin{proposition}
\label{Prop.T1p}If $\psi\in\mathcal{E}_{\ast}\left(  \mathbb{R}^{n}\right)  $
is homogeneous of order $k\in\mathbb{Z},$ then%
\begin{equation}
\frac{\partial^{\ast}}{\partial x_{i}}\left(  \mathcal{P}f\left(  \psi\right)
\right)  =\mathcal{P}f\left(  \frac{\partial\psi}{\partial x_{i}}\right)
+Cn_{i}\psi_{0}\delta_{\ast}^{\left[  1-n-k\right]  }\,, \label{T1.4}%
\end{equation}
where $\psi_{0}\left(  \mathbf{x}\right)  =\left\vert \mathbf{x}\right\vert
^{-k}\psi\left(  \mathbf{x}\right)  .$\smallskip
\end{proposition}

If we now apply the projection $\Pi$ onto the usual distribution space
$\mathcal{D}^{\prime}\left(  \mathbb{R}^{n}\right)  ,$ we obtain the formula
for the distributional derivatives of homogeneous distributions. Observe first
that if $k>-n$ then $\psi$ is integrable at the origin, and thus $\psi$ is a
regular distribution and $\Pi\left(  \mathcal{P}f\left(  \psi\right)  \right)
=\psi.$ If $k\leq-n$ then $\Pi\left(  \mathcal{P}f\left(  \psi\right)
\right)  =\mathcal{P}f\left(  \psi\right)  ,$ since in that case the integral
$\int_{\mathbb{R}^{n}}\psi\left(  \mathbf{x}\right)  \phi\left(
\mathbf{x}\right)  \,\mathrm{d}\mathbf{x}$ would be divergent, in general, if
$\phi\in\mathcal{D}\left(  \mathbb{R}^{n}\right)  .$ A particularly
interesting case is when $k=-n,$ since if $\psi$ is homogeneous of degree $-n$
and%
\begin{equation}
\int_{\mathbb{S}}\psi\left(  \mathbf{w}\right)  \,\mathrm{d}\sigma\left(
\mathbf{w}\right)  =0\,, \label{T1.4p}%
\end{equation}
then the \emph{principal value} of the integral%
\begin{equation}
\mathrm{p.v.}\int_{\mathbb{R}^{n}}\psi\left(  \mathbf{x}\right)  \phi\left(
\mathbf{x}\right)  \,\mathrm{d}\mathbf{x=}\lim_{\varepsilon\rightarrow0^{+}%
}\int_{\left\vert \mathbf{x}\right\vert \geq\varepsilon}\psi\left(
\mathbf{x}\right)  \phi\left(  \mathbf{x}\right)  \,\mathrm{d}\mathbf{x}\,,
\label{T1.4p2}%
\end{equation}
actually exists for each $\phi\in\mathcal{D}\left(  \mathbb{R}^{n}\right)  ,$
so that $\mathcal{P}f\left(  \psi\right)  =\mathrm{p.v.}\left(  \psi\right)
,$ the principal value distribution\footnote{Let $\Sigma$ be a closed surface
in $\mathbb{R}^{n}$ that encloses the origin. We describe $\Sigma$ by an
equation of the form $g(\mathbf{x})=1,$ where $g(\mathbf{x})$ is continuous in
$\mathbb{R}^{n}\setminus\left\{  \mathbf{0}\right\}  $ and homogeneous of
degree $1.$ Then $\left\langle \mathcal{R}_{\Sigma}\left(  \psi(\mathbf{x}%
)\right)  ,\phi(\mathbf{x})\right\rangle =\underset{\varepsilon\rightarrow
0}{\lim}\int_{g(\mathbf{x}\dot{)}\geqslant\varepsilon}\psi(\mathbf{x}%
)\phi(\mathbf{x})\,\mathrm{d}\mathbf{x}\,,$ defines another regularization of
$\psi,$ but in general $\mathcal{R}_{\Sigma}\left(  \psi(\mathbf{x})\right)
\neq\mathrm{p.v.}\left(  \psi\left(  \mathbf{x}\right)  \right)  $ \cite{YE},
a fact observed by Farassat \cite{Farassat}, who indicated its importance in
numerical computations, and studied by several authors \cite{VH1, YE}.}.
Condition (\ref{T1.4p}) holds whenever $\psi=\partial\xi/\partial x_{j}$ for
some $\xi$ homogeneous of order $-n+1.$\smallskip

\begin{proposition}
\label{Prop. T1.2}Let $\psi$ be homogeneous of order $k\in\mathbb{Z}$ in
$\mathbb{R}^{n}\setminus\left\{  \mathbf{0}\right\}  .$ Then, in
$\mathcal{D}^{\prime}\left(  \mathbb{R}^{n}\right)  $ the distributional
derivative $\overline{\partial}\psi/\partial x_{i}$ is given as follows:%
\begin{equation}
\frac{\overline{\partial}\psi}{\partial x_{i}}=\frac{\partial\psi}{\partial
x_{i}}\,,\ \ \ \ \ k>1-n\,, \label{T1.5}%
\end{equation}
equality of regular distributions;%
\begin{equation}
\frac{\overline{\partial}\psi}{\partial x_{i}}=\mathrm{p.v.}\left(
\frac{\partial\psi}{\partial x_{i}}\right)  +A\delta\left(  \mathbf{x}\right)
\,,\ \ \ \ k=1-n\,, \label{T1.6}%
\end{equation}
where $A=\int_{\mathbb{S}}n_{i}\psi_{0}\left(  \mathbf{w}\right)
\,\mathrm{d}\sigma\left(  \mathbf{w}\right)  =\left\langle \psi_{0}%
,n_{i}\right\rangle _{\mathcal{D}^{\prime}\left(  \mathbb{S}\right)
\times\mathcal{D}\left(  \mathbb{S}\right)  },$ while%
\begin{equation}
\frac{\overline{\partial}\psi}{\partial x_{i}}=\mathcal{P}f\left(
\frac{\partial\psi}{\partial x_{i}}\right)  +D\left(  \mathbf{x}\right)
\,,\ \ \ \ k<1-n\,, \label{T1.7}%
\end{equation}
where $D\left(  \mathbf{x}\right)  $ is a homogeneous distribution of order
$k-1$ concentrated at the origin and given by%
\begin{equation}
D\left(  \mathbf{x}\right)  =\left(  -1\right)  ^{-k-n+1}\sum_{j_{1}%
+\cdots+j_{n}=-k-n+1}\frac{\left\langle n_{i}\psi_{0},\mathbf{w}^{\left(
j_{1},\ldots,j_{n}\right)  }\right\rangle }{j_{1}!\cdots j_{n}!}%
\mathbf{D}^{\left(  j_{1},\ldots,j_{n}\right)  }\delta\left(  \mathbf{x}%
\right)  \,. \label{T1.8}%
\end{equation}

\end{proposition}

\begin{proof}
It follows from (\ref{T1.3}) if we observe \cite[Prop. 4.7]{YE2} that if
$g\in\mathcal{D}^{\prime}\left(  \mathbb{S}\right)  $ then%
\begin{equation}
\Pi\left(  g\delta_{\ast}^{\left[  q\right]  }\right)  =\frac{\left(
-1\right)  ^{q}}{C}\sum_{j_{1}+\cdots+j_{n}=q}\frac{\left\langle g\left(
\mathbf{w}\right)  ,\mathbf{w}^{\left(  j_{1},\ldots,j_{n}\right)
}\right\rangle }{j_{1}!\cdots j_{n}!}\mathbf{D}^{\left(  j_{1},\ldots
,j_{n}\right)  }\delta\left(  \mathbf{x}\right)  \,, \label{T1.9}%
\end{equation}
and, in particular,
\begin{equation}
\Pi\left(  g\delta_{\ast}\right)  =\frac{1}{C}\left\langle g\left(
\mathbf{w}\right)  ,1\right\rangle \delta\left(  \mathbf{x}\right)  \,,
\label{T1.10}%
\end{equation}
if $q=0.$\smallskip
\end{proof}

Our next task is to compute the second order thick derivatives of homogeneous
distributions. Indeed, if $\psi$ is homogeneous of degree $k$ then we can
iterate the formula (\ref{T1.4}) to obtain%
\begin{align}
&  \frac{\partial^{\ast2}}{\partial x_{i}\partial x_{j}}\left(  \mathcal{P}%
f\left(  \psi\right)  \right)  =\frac{\partial^{\ast}}{\partial x_{i}}\left(
\mathcal{P}f\left(  \frac{\partial\psi}{\partial x_{j}}\right)  +Cn_{j}%
\psi_{0}\delta_{\ast}^{\left[  1-n-k\right]  }\right) \label{T1.11}\\
&  \ \ \ \ \ \ \ \ \ =\mathcal{P}f\left(  \frac{\partial^{2}\psi}{\partial
x_{i}\partial x_{j}}\right)  +Cn_{i}\xi_{0}\delta_{\ast}^{\left[
2-n-k\right]  }+\frac{\partial^{\ast}}{\partial x_{i}}\left(  Cn_{j}\psi
_{0}\delta_{\ast}^{\left[  1-n-k\right]  }\right)  \,,\nonumber
\end{align}
where $\xi=\partial\psi/\partial x_{j}$ is homogeneous of degree $k-1$ and
$\xi_{0}\left(  \mathbf{x}\right)  =\left\vert \mathbf{x}\right\vert ^{1-k}%
\xi\left(  \mathbf{x}\right)  $ is the associated homogeneous of degree $0$
function. Use of (\ref{DD.1}) and (\ref{Su.5}) allows us to write%
\begin{align}
&  \frac{\partial^{\ast}}{\partial x_{i}}\left(  Cn_{j}\psi_{0}\delta_{\ast
}^{\left[  1-n-k\right]  }\right)  =C\left(  \frac{\delta}{\delta x_{i}%
}\left(  n_{j}\psi_{0}\right)  +\left(  k-1\right)  n_{i}n_{j}\psi_{0}\right)
\delta_{\ast}^{\left[  2-n-k\right]  }\label{T1.12}\\
&  \ \ \ \ \ =C\left(  \left(  \delta_{ij}-n_{i}n_{j}\right)  \psi_{0}%
+n_{j}\frac{\delta\psi_{0}}{\delta x_{i}}+\left(  k-1\right)  n_{i}n_{j}%
\psi_{0}\right)  \delta_{\ast}^{\left[  2-n-k\right]  }\nonumber\\
&  \ \ \ \ \ =C\left(  \left(  \delta_{ij}+\left(  k-2\right)  n_{i}%
n_{j}\right)  \psi_{0}+n_{j}\frac{\delta\psi_{0}}{\delta x_{i}}\right)
\delta_{\ast}^{\left[  2-n-k\right]  }\,,\nonumber
\end{align}
while the equation $\psi=r^{k}\psi_{0}$ yields $\partial\psi/\partial
x_{j}=r^{k-1}\{kn_{j}\psi_{0}+\delta\psi_{0}/\delta x_{j}\},$ so that%
\begin{equation}
\xi_{0}=kn_{j}\psi_{0}+\frac{\delta\psi_{0}}{\delta x_{j}}\,. \label{T1.13}%
\end{equation}
Collecting terms we thus obtain the following formula.\smallskip

\begin{proposition}
\label{Prop. T1. 4}If $\psi\in\mathcal{E}_{\ast}\left(  \mathbb{R}^{n}\right)
$ is homogeneous of order $k\in\mathbb{Z},$ then%
\begin{align}
&  \frac{\partial^{\ast2}}{\partial x_{i}\partial x_{j}}\left(  \mathcal{P}%
f\left(  \psi\right)  \right)  =\mathcal{P}f\left(  \frac{\partial^{2}\psi
}{\partial x_{i}\partial x_{j}}\right) \label{T1.14}\\
&  \ \ \ \ \ +C\left(  \left(  \delta_{ij}+2\left(  k-1\right)  n_{i}%
n_{j}\right)  \psi_{0}+n_{j}\frac{\delta\psi_{0}}{\delta x_{i}}+n_{i}%
\frac{\delta\psi_{0}}{\delta x_{j}}\right)  \delta_{\ast}^{\left[
2-n-k\right]  }\,.\nonumber
\end{align}
where $\psi_{0}\left(  \mathbf{x}\right)  =\left\vert \mathbf{x}\right\vert
^{-k}\psi\left(  \mathbf{x}\right)  .$\smallskip
\end{proposition}

Projection onto $\mathcal{D}^{\prime}\left(  \mathbb{R}^{n}\right)  $ of
(\ref{T1.14}) gives the formula for the distributional derivatives
$\overline{\partial}^{2}/\partial x_{i}\partial x_{j}(\mathcal{P}f\left(
\psi\right)  )$ if $\psi\in\mathcal{E}_{\ast}\left(  \mathbb{R}^{n}\right)  $
is homogeneous of order $k\in\mathbb{Z}.$ In case $k=2-n$ we obtain the
following formula.\smallskip

\begin{proposition}
\label{Prop. T1.5}If $\psi\in\mathcal{E}_{\ast}\left(  \mathbb{R}^{n}\right)
$ is homogeneous of order $2-n,$ then%
\begin{equation}
\frac{\overline{\partial}^{2}}{\partial x_{i}\partial x_{j}}\left(
\psi\right)  =\mathrm{p.v.}\left(  \frac{\partial^{2}\psi}{\partial
x_{i}\partial x_{j}}\right)  +B\delta\left(  \mathbf{x}\right)  \,,
\label{T1.15}%
\end{equation}
where%
\begin{equation}
B=\left\langle \psi_{0},2n_{i}n_{j}-\delta_{ij}\right\rangle _{\mathcal{D}%
^{\prime}\left(  \mathbb{S}\right)  \times\mathcal{D}\left(  \mathbb{S}%
\right)  }\,. \label{T1.16}%
\end{equation}

\end{proposition}

\begin{proof}
If we apply the operator $\Pi$ to (\ref{T1.14}) and employ (\ref{T1.10}) we
obtain (\ref{T1.15}) with%
\[
B=\left\langle \left(  \delta_{ij}+2\left(  k-1\right)  n_{i}n_{j}\right)
\psi_{0}+n_{j}\frac{\delta\psi_{0}}{\delta x_{j}}+n_{i}\frac{\delta\psi_{0}%
}{\delta x_{j}},1\right\rangle _{\mathcal{D}^{\prime}\left(  \mathbb{S}%
\right)  \times\mathcal{D}\left(  \mathbb{S}\right)  }.
\]
But \cite[(2.6)]{YE2} yields%
\begin{equation}
\left\langle n_{j}\frac{\delta\psi_{0}}{\delta x_{j}},1\right\rangle
_{\mathcal{D}^{\prime}\left(  \mathbb{S}\right)  \times\mathcal{D}\left(
\mathbb{S}\right)  }=\left\langle \psi_{0},n\,n_{i}n_{j}-\delta_{ij}%
\right\rangle _{\mathcal{D}^{\prime}\left(  \mathbb{S}\right)  \times
\mathcal{D}\left(  \mathbb{S}\right)  }, \label{T1.17}%
\end{equation}
and (\ref{T1.16}) follows since $k=2-n.$\smallskip
\end{proof}

We would like to observe that while $\psi_{0}$ has been supposed smooth, a
continuity argument immediately gives that $\psi_{0}$ could be any
distribution of $\mathcal{D}^{\prime}\left(  \mathbb{R}^{n}\setminus\left\{
\mathbf{0}\right\}  \right)  $ that is homogeneous of degree $0.$

\section{Bowen's formula\label{Section; Bowen}}

If we apply formula (\ref{T1.6}) to the function $\psi=n_{j_{1}}\cdots
n_{j_{k}}/r^{2},$ which is homogeneous of degree $-2$ in $\mathbb{R}^{3}$ we
obtain at once that
\begin{align}
&  \frac{\overline{\partial}}{\partial x_{i}}\left(  \frac{n_{j_{1}}\cdots
n_{j_{k}}}{r^{2}}\right)  =\label{B.1}\\
&  \mathrm{p.v.}\left(  \left\{  \sum_{q=1}^{k}\delta_{ij_{q}}\frac{n_{j_{1}%
}\cdots n_{j_{k}}}{n_{j_{q}}}-\left(  k+2\right)  n_{i}n_{j_{1}}\cdots
n_{j_{k}}\right\}  \frac{1}{r^{3}}\right)  +A\delta\left(  \mathbf{x}\right)
\,,\nonumber
\end{align}
where
\begin{equation}
A=\int_{\mathbb{S}}n_{i}n_{j_{1}}\cdots n_{j_{k}}\,\mathrm{d}\sigma\left(
\mathbf{w}\right)  \,. \label{B.2}%
\end{equation}
This integral was computed in \cite[(3.13)]{EK85b}, the result being
\begin{equation}
A=\frac{2\Gamma\left(  \left(  a+1\right)  /2\right)  \Gamma\left(  \left(
b+1\right)  /2\right)  \Gamma\left(  \left(  c+1\right)  /2\right)  }%
{\Gamma\left(  \left(  a+b+c+3\right)  /2\right)  }\,, \label{B.3}%
\end{equation}
if $n_{i}n_{j_{1}}\cdots n_{j_{k}}=n_{1}^{a}n_{2}^{b}n_{3}^{c},$ and $a,$ $b,$
or $c$ are even, while $A=0$ if any exponent is odd. Bowen \cite[Eqn.
(A5)]{Bowen} also computes the integral, and obtains a different but
equivalent expression; in particular, his formula for $k=3$ reads as%
\begin{equation}
A=\frac{4\pi}{15}\left(  \delta_{ij_{1}}\delta_{j_{2}j_{3}}+\delta_{ij_{2}%
}\delta_{j_{1}j_{3}}+\delta_{ij_{3}}\delta_{j_{1}j_{2}}\right)  \,,
\label{B.3p}%
\end{equation}
so that (\ref{B.3}) or (\ref{B.3p}) would yield that if $\left(  a,b,c\right)
$ is a permutation of $\left(  2,2,0\right)  $ then $A=4\pi/15$ while if a
permutation of $\left(  4,0,0\right)  $ then $A=4\pi/5.$

Our main aim is to point out why the product rule for derivatives, as employed
in \cite{Bowen} does not produce the correct result. Indeed, if we use
\cite[Eqn. (16)]{Bowen} written as\footnote{We shall employ our notation, not
the original one of \cite{Bowen}.}%
\begin{equation}
\frac{\overline{\partial}}{\partial x_{i}}\left(  \frac{n_{j_{1}}}{r^{2}%
}\right)  =\mathrm{p.v.}\left(  \frac{\delta_{ij_{1}}-3n_{i}n_{j_{i}}}{r^{3}%
}\right)  +\frac{4\pi}{3}\delta_{ij_{1}}\delta\left(  \mathbf{x}\right)  \,,
\label{B.4}%
\end{equation}
and then try to proceed as in \cite[Eqn. (18)]{Bowen},%
\begin{equation}
\frac{\overline{\partial}}{\partial x_{i}}\left(  \frac{n_{j_{1}}n_{j_{2}%
}n_{j_{3}}}{r^{2}}\right)  \text{\textquotedblleft\textquestiondown }%
=\text{?\textquotedblright\ }n_{j_{1}}n_{j_{2}}\frac{\overline{\partial}%
}{\partial x_{i}}\left(  \frac{n_{j_{3}}}{r^{2}}\right)  +\frac{n_{j_{3}}%
}{r^{2}}\frac{\overline{\partial}}{\partial x_{i}}\left(  n_{j_{1}}n_{j_{2}%
}\right)  \,. \label{B.5}%
\end{equation}
Thus (\ref{B.4}) and the formula%
\begin{equation}
\frac{\overline{\partial}}{\partial x_{i}}\left(  n_{j_{1}}n_{j_{2}}\right)
=\frac{\delta_{ij_{1}}n_{j_{2}}+\delta_{ij_{2}}n_{j_{1}}-2n_{i}n_{j_{1}%
}n_{j_{2}}}{r}\,, \label{B.6}%
\end{equation}
give%
\begin{equation}
n_{j_{1}}n_{j_{2}}\frac{\overline{\partial}}{\partial x_{i}}\left(
\frac{n_{j_{3}}}{r^{2}}\right)  +\frac{n_{j_{3}}}{r^{2}}\frac{\overline
{\partial}}{\partial x_{i}}\left(  n_{j_{1}}n_{j_{2}}\right)
=\text{\textquotedblleft Normal\textquotedblright}+\text{\textquotedblleft
Src\textquotedblright}\,, \label{B.7}%
\end{equation}
where%
\begin{equation}
\text{\textquotedblleft Normal\textquotedblright}=\mathrm{p.v.}\left(
\frac{\delta_{ij_{1}}n_{j_{2}}n_{j_{3}}+\delta_{ij_{2}}n_{j_{1}}n_{j_{3}%
}+\delta_{ij_{3}}n_{j_{1}}n_{j_{2}}-5n_{i}n_{j_{1}}n_{j_{2}}n_{j_{3}}}{r^{3}%
}\right)  \,, \label{B.8}%
\end{equation}
coincides with the first term of (\ref{B.1}) while
\begin{equation}
\text{\textquotedblleft Src\textquotedblright}=\frac{4\pi}{3}\delta_{ij_{3}%
}n_{j_{1}}n_{j_{2}}\delta\left(  \mathbf{x}\right)  \,. \label{B.9}%
\end{equation}
The right hand side of (\ref{B.9}) is not a well defined distribution, of
course, but Bowen suggested that we treat it as what we now call the
projection of a thick distribution, that is, as
\begin{equation}
\text{\textquotedblleft Src\textquotedblright}=\Pi\left(  \frac{4\pi}{3}%
\delta_{ij_{3}}n_{j_{1}}n_{j_{2}}\delta_{\ast}\right)  =\frac{4\pi}{9}%
\delta_{ij_{3}}\delta_{j_{1}j_{2}}\delta\left(  \mathbf{x}\right)  \,,
\label{B.10}%
\end{equation}
since $\Pi\left(  n_{j_{1}}n_{j_{2}}\delta_{\ast}\right)  =\left(  1/3\right)
\delta_{j_{1}j_{2}}\delta\left(  \mathbf{x}\right)  $ \cite[Example 5.10]%
{YE2}. In order to compare with (\ref{B.1}) and (\ref{B.3p}) we observe that
by symmetry the same result would be obtained if $j_{3}$ and $j_{1},$ or
$j_{3}$ and $j_{2},$ are exchanged, so that if in the term \textquotedblleft
Src\textquotedblright\ we do these exchanges, add the results and divide by
$3,$ we would get
\begin{equation}
\text{\textquotedblleft SrcSym\textquotedblright}=\frac{4\pi}{27}\left(
\delta_{ij_{1}}\delta_{j_{2}j_{3}}+\delta_{ij_{2}}\delta_{j_{1}j_{3}}%
+\delta_{ij_{3}}\delta_{j_{1}j_{2}}\right)  \delta\left(  \mathbf{x}\right)
\,, \label{B.11}%
\end{equation}
and thus the symmetric version of the (\ref{B.7}) is \textquotedblleft
Normal\textquotedblright$+$\textquotedblleft SrcSym\textquotedblright, which
of course is different from (\ref{B.1}) since the coefficient in (\ref{B.3p})
is $4\pi/15,$ while that in (\ref{B.11}) is $4\pi/27.$ Therefore, the relation
\textquotedblleft\textquestiondown $=$?\textquotedblright\ in (\ref{B.5})
cannot be replaced by $=.$

Hence the product rule for derivatives fails in this case. \emph{The question
is why? }Indeed, when computing the right side of (\ref{B.5}), that is, the
left side of (\ref{B.7}), we found just one irregular product, namely
$n_{j_{1}}n_{j_{2}}\delta\left(  \mathbf{x}\right)  ,$ but using the average
value $\left(  1/3\right)  \delta_{j_{1}j_{2}}\delta\left(  \mathbf{x}\right)
$ seems quite reasonable.

In order to see what went wrong let us compute $\overline{\partial}/\partial
x_{i}\left(  n_{j_{1}}n_{j_{2}}n_{j_{3}}/r^{2}\right)  $ by computing the
thick derivative $\partial^{\ast}/\partial x_{i}\mathcal{P}f\left(  n_{j_{1}%
}n_{j_{2}}n_{j_{3}}/r^{2}\right)  ,$ applying the product rule for thick
derivatives, and then taking the projection $\pi$ of this. We have,%
\begin{align*}
\frac{\partial^{\ast}}{\partial x_{i}}\mathcal{P}f\left(  \frac{n_{j_{1}%
}n_{j_{2}}n_{j_{3}}}{r^{2}}\right)   &  =\frac{\partial^{\ast}}{\partial
x_{i}}\left[  n_{j_{1}}n_{j_{2}}\mathcal{P}f\left(  \frac{n_{j_{3}}}{r^{2}%
}\right)  \right] \\
&  =n_{j_{1}}n_{j_{2}}\frac{\partial^{\ast}}{\partial x_{i}}\mathcal{P}%
f\left(  \frac{n_{j_{3}}}{r^{2}}\right)  +\frac{\partial\left(  n_{j_{1}%
}n_{j_{2}}\right)  }{\partial x_{i}}\mathcal{P}f\left(  \frac{n_{j_{3}}}%
{r^{2}}\right)  \,,
\end{align*}
and taking (\ref{T1.4}) into account, we obtain
\begin{align*}
&  n_{j_{1}}n_{j_{2}}\left\{  \mathcal{P}f\left(  \frac{\delta_{ij_{3}}%
-3n_{i}n_{j_{3}}}{r^{3}}\right)  +4\pi n_{j_{3}}n_{i}\delta_{\ast}\right\} \\
&  \ \ \ \ \ \ \ \ \ \ \ \ \ \ \ \ \ \ \ \ \ \ \ \ +\frac{\delta_{ij_{1}%
}n_{j_{2}}+\delta_{ij_{2}}n_{j_{1}}-2n_{i}n_{j_{1}}n_{j_{2}}}{r}%
\mathcal{P}f\left(  \frac{n_{j_{3}}}{r^{2}}\right)  \,,
\end{align*}
that is, $\partial^{\ast}/\partial x_{i}\mathcal{P}f\left(  n_{j_{1}}n_{j_{2}%
}n_{j_{3}}/r^{2}\right)  $ equals%
\begin{align}
&  \mathcal{P}f\left(  \frac{\delta_{ij_{1}}n_{j_{2}}n_{j_{3}}+\delta_{ij_{2}%
}n_{j_{1}}n_{j_{3}}+\delta_{ij_{3}}n_{j_{1}}n_{j_{2}}-5n_{i}n_{j_{1}}n_{j_{2}%
}n_{j_{3}}}{r^{3}}\right) \label{B.12}\\
&  \ \ \ \ \ \ \ \ \ \ \ \ \ \ \ \ \ \ \ \ \ \ \ \ +4\pi n_{j_{1}}n_{j_{2}%
}n_{j_{3}}n_{i}\delta_{\ast}\,.\nonumber
\end{align}
Applying the projection operator $\Pi$ we obtain that the $\mathcal{P}f$
becomes a $\mathrm{p.v.}$, so that the term \textquotedblleft
Normal\textquotedblright\ given by (\ref{B.8}) is obtained, while
(\ref{T1.10}) yields that the projection of thick delta is exactly
$A\delta\left(  \mathbf{x}\right)  $ where $A=\int_{\mathbb{S}}n_{i}n_{j_{1}%
}n_{j_{2}}n_{j_{3}}\,\mathrm{d}\sigma\left(  \mathbf{w}\right)  ,$ that is,
the \emph{correct} term
\[
\frac{4\pi}{15}\left(  \delta_{ij_{1}}\delta_{j_{2}j_{3}}+\delta_{ij_{2}%
}\delta_{j_{1}j_{3}}+\delta_{ij_{3}}\delta_{j_{1}j_{2}}\right)  \delta\left(
\mathbf{x}\right)  \,.
\]
The reason we now obtain the correct result is while it is true that
$\Pi\left(  n_{j_{1}}n_{j_{2}}\delta_{\ast}\right)  =\left(  1/3\right)
\delta_{j_{1}j_{2}}\delta\left(  \mathbf{x}\right)  $ and that $\Pi\left(
n_{j_{3}}n_{i}\delta_{\ast}\right)  =\left(  1/3\right)  \delta_{ij_{3}}%
\delta\left(  \mathbf{x}\right)  ,$ it is \emph{not} true that the projection
$\Pi\left(  4\pi n_{j_{1}}n_{j_{2}}n_{j_{3}}n_{i}\delta_{\ast}\right)  $ can
be obtained as $4\pi\left(  1/3\right)  \delta_{ij_{3}}\Pi\left(  n_{j_{1}%
}n_{j_{2}}\delta_{\ast}\right)  $ nor as $4\pi\left(  1/3\right)
\delta_{j_{1}j_{2}}\Pi\left(  n_{j_{3}}n_{i}\delta_{\ast}\right)  ,$ and
actually not even the symmetrization of such results, given by (\ref{B.11}),
works. Put in simple terms, it is not true that the average of a product is
the product of the averages!

One can, alternatively, compute $\partial^{\ast}/\partial x_{i}\mathcal{P}%
f\left(  n_{j_{1}}n_{j_{2}}n_{j_{3}}/r^{2}\right)  $ as%
\begin{equation}
\frac{\partial}{\partial x_{i}}\left(  \frac{n_{j_{3}}}{r^{2}}\right)
\mathcal{P}f\left(  n_{j_{1}}n_{j_{2}}\right)  +\left(  \frac{n_{j_{3}}}%
{r^{2}}\right)  \frac{\partial^{\ast}}{\partial x_{i}}\mathcal{P}f\left(
n_{j_{1}}n_{j_{2}}\right)  \,, \label{B.13}%
\end{equation}
since%
\begin{equation}
\frac{\partial^{\ast}}{\partial x_{i}}\mathcal{P}f\left(  n_{j_{1}}n_{j_{2}%
}\right)  =\mathcal{P}f\left(  \frac{\delta_{ij_{1}}n_{j_{2}}+\delta_{ij_{2}%
}n_{j_{1}}-2n_{i}n_{j_{1}}n_{j_{2}}}{r}\right)  +4\pi n_{j_{1}}n_{j_{2}}%
n_{i}\delta_{\ast}^{\left[  -2\right]  }\,. \label{B.14}%
\end{equation}
Here the thick delta term in (\ref{B.13})\ is $4\pi\left(  n_{j_{3}}%
/r^{2}\right)  n_{j_{1}}n_{j_{2}}n_{i}\delta_{\ast}^{\left[  -2\right]  },$
which becomes, as it should, $4\pi n_{j_{1}}n_{j_{2}}n_{j_{3}}n_{i}%
\delta_{\ast}.$

Complications in the use of the product rule for derivatives in one variable
were considered in \cite{EF07} when analysing the formula \cite{Vibet}%
\begin{equation}
\frac{d}{dx}\left(  H^{n}\left(  x\right)  \right)  =nH^{n-1}\left(  x\right)
\delta\left(  x\right)  \,, \label{B.15}%
\end{equation}
\ where $H$ is the Heaviside function; see also \cite{Paskusz}.

\section{Higher order derivatives\label{Section: Higher order derivatives}}

We now consider the computation of higher order derivatives in the space
$\left(  \mathcal{D}_{\ast}^{\left[  0\right]  }\left(  \mathbb{R}^{n}\right)
\right)  ^{\prime}.$ If $f\in\mathcal{D}_{\ast}^{\prime}\left(  \mathbb{R}%
^{n}\right)  $ then, of course, the thick derivative $\partial^{\ast
}f/\partial x_{i}$ is defined by duality, that is,
\begin{equation}
\left\langle \frac{\partial^{\ast}f}{\partial x_{i}},\phi\right\rangle
=-\left\langle f,\frac{\partial\phi}{\partial x_{i}}\right\rangle \,,
\label{h.1}%
\end{equation}
for $\phi\in\mathcal{D}_{\ast}\left(  \mathbb{R}^{n}\right)  .$ Suppose now
that $\mathcal{A}$ is a subspace of $\mathcal{D}_{\ast}\left(  \mathbb{R}%
^{n}\right)  $ that has a topology such that the imbedding $i:\mathcal{A}%
\hookrightarrow\mathcal{D}_{\ast}\left(  \mathbb{R}^{n}\right)  $ is
continuous; then the transpose $i^{T}:\mathcal{D}_{\ast}^{\prime}\left(
\mathbb{R}^{n}\right)  \rightarrow\mathcal{A}^{\prime}$ is just the
restriction operator $\Pi_{\mathcal{A}}.$ If $\mathcal{A}$ is closed under the
differentiation operators\footnote{The space $\mathcal{A}^{\prime}$ would be a
space of (thick) distributions in the sense of Zemanian \cite{zem}.}, then we
can also define the derivative of any $f\in\mathcal{A}^{\prime},$ say
$\partial_{\mathcal{A}}f/\partial x_{i},$ by employing (\ref{h.1}) for
$\phi\in\mathcal{A}.$ Then%
\begin{equation}
\Pi_{\mathcal{A}}\left(  \frac{\partial^{\ast}f}{\partial x_{i}}\right)
=\frac{\partial_{\mathcal{A}}}{\partial x_{i}}\left(  \Pi_{\mathcal{A}}\left(
f\right)  \right)  \,, \label{h.2}%
\end{equation}
for any thick distribution $f\in\mathcal{D}_{\ast}^{\prime}\left(
\mathbb{R}^{n}\right)  .$ In the particular case when $\mathcal{A}%
=\mathcal{D}\left(  \mathbb{R}^{n}\right)  $ then $\partial_{\mathcal{A}%
}f/\partial x_{i}=\overline{\partial}f/\partial x_{i},$ the usual
distributional derivative, and thus (\ref{h.2}) becomes \cite[Eqn.
(5.22)]{YE2},%
\begin{equation}
\Pi\left(  \frac{\partial^{\ast}f}{\partial x_{i}}\right)  =\frac
{\overline{\partial}\Pi\left(  f\right)  }{\partial x_{i}}\,. \label{h.3}%
\end{equation}
What this means is that one can use thick distributional derivatives to
compute $\partial_{\mathcal{A}}f/\partial x_{i},$ as we have already done to
compute distributional derivatives.

When $\mathcal{A}$ is not closed under the differentiation operators then
$\partial_{\mathcal{A}}f/\partial x_{i}$ cannot be defined by (\ref{h.1}) if
$f\in\mathcal{A}^{\prime}$ since in general $\partial\phi/\partial x_{i}$ does
not belong to $\mathcal{A}$ and thus the right side of (\ref{h.1}) is not
defined. However, if $f\in\mathcal{A}^{\prime}$ has a \emph{canonical}
extension $\widetilde{f}\in\mathcal{D}_{\ast}^{\prime}\left(  \mathbb{R}%
^{n}\right)  $ then we could define $\partial_{\mathcal{A}}f/\partial x_{i}$
as $\Pi_{\mathcal{A}}\left(  \partial^{\ast}\widetilde{f}/\partial
x_{i}\right)  .$ This applies, in particular when $\mathcal{A}=\mathcal{D}%
_{\ast}^{\left[  0\right]  }\left(  \mathbb{R}^{n}\right)  :$ if $f\in\left(
\mathcal{D}_{\ast}^{\left[  0\right]  }\left(  \mathbb{R}^{n}\right)  \right)
^{\prime}$ then $\partial_{0}^{\ast}f/\partial x_{i}=\partial_{\mathcal{A}%
}f/\partial x_{i}$ cannot be defined, in general, but if $f$ has a canonical
extension $\widetilde{f}\in\mathcal{D}_{\ast}^{\prime}\left(  \mathbb{R}%
^{n}\right)  $\ then $\partial_{0}^{\ast}f/\partial x_{i}$ is understood as
$\Pi_{\mathcal{D}_{\ast}^{\left[  0\right]  }\left(  \mathbb{R}^{n}\right)
}\left(  \partial^{\ast}\widetilde{f}/\partial x_{i}\right)  .$

Our aim is to point out that, in general, if $P=RS$ is the product of two
differential operators with constant coefficients, then while, with obvious
notations, $P^{\ast}=R^{\ast}S^{\ast},$ $P_{\mathcal{A}}=R_{\mathcal{A}%
}S_{\mathcal{A}},$ if $\mathcal{A}$ is closed under differential operators,
and $\overline{P}=\overline{R}\,\overline{S},$ it is \emph{not} true that
$P_{0}^{\ast}=R_{0}^{\ast}S_{0}^{\ast}.$ Therefore the space $\left(
\mathcal{D}_{\ast}^{\left[  0\right]  }\left(  \mathbb{R}^{n}\right)  \right)
^{\prime}$ is not a convenient framework to generalize distributions to thick
distributions; the whole $\mathcal{D}_{\ast}^{\prime}\left(  \mathbb{R}%
^{n}\right)  $ is needed if we want a theory that includes the possibility of
differentiation.\smallskip

\begin{example}
\label{Example h.1}Let us consider the second order derivatives of the
distribution $\mathcal{P}f\left(  1\right)  .$ Formula (\ref{T1.14}) yields%
\begin{equation}
\frac{\partial^{\ast2}}{\partial x_{i}\partial x_{j}}\left(  \mathcal{P}%
f\left(  1\right)  \right)  =C\left(  \delta_{ij}-2n_{i}n_{j}\right)
\delta_{\ast}^{\left[  -n+2\right]  }. \label{H.4}%
\end{equation}
In particular, in $\mathbb{R}^{2},$ $\partial^{\ast2}/\partial x_{i}\partial
x_{j}\left(  \mathcal{P}f\left(  1\right)  \right)  =2\pi\left(  \delta
_{ij}-2n_{i}n_{j}\right)  \delta_{\ast}.$ If we consider the function $1$ as
an element of $\left(  \mathcal{D}_{\ast}^{\left[  0\right]  }\left(
\mathbb{R}^{2}\right)  \right)  ^{\prime}$ then it has the canonical extension
$\mathcal{P}f\left(  1\right)  \in\mathcal{D}_{\ast}^{\prime}\left(
\mathbb{R}^{2}\right)  $ and so%
\[
\frac{\partial_{0}^{\ast}\left(  1\right)  }{\partial x_{j}}=\Pi
_{\mathcal{D}_{\ast}^{\left[  0\right]  }\left(  \mathbb{R}^{2}\right)
}\left(  2\pi n_{j}\delta_{\ast}^{\left[  -1\right]  }\right)  =0\,,
\]
and consequently,%
\begin{equation}
\frac{\partial_{0}^{\ast}}{\partial x_{i}}\left(  \frac{\partial_{0}^{\ast
}\left(  1\right)  }{\partial x_{j}}\right)  =\frac{\partial_{0}^{\ast}%
}{\partial x_{i}}\left(  0\right)  =0\neq2\pi\left(  \delta_{ij}-2n_{i}%
n_{j}\right)  \delta_{\ast}=\frac{\partial^{\ast2}\left(  1\right)  }{\partial
x_{i}\partial x_{j}}\,. \label{h.5}%
\end{equation}
Observe that $\Pi\left(  2\pi\left(  \delta_{ij}-2n_{i}n_{j}\right)
\delta_{\ast}\right)  =0,$ but observe also that this means very
little.\smallskip
\end{example}

\begin{example}
\label{Example h.2}It was obtained in \cite[Thm. 7.6]{YE2} that in
$\mathcal{D}_{\ast}^{\prime}\left(  \mathbb{R}^{3}\right)  $
\begin{equation}
\frac{\partial^{\ast2}\mathcal{P}f\left(  r^{-1}\right)  }{\partial
x_{i}\partial x_{j}}=\left(  3x_{i}x_{j}-\delta_{ij}r^{2}\right)
\mathcal{P}f\left(  r^{-5}\right)  +4\pi\left(  \delta_{ij}-4n_{i}%
n_{j}\right)  \delta_{\ast}\,. \label{PD.9}%
\end{equation}
Since $\Pi\left(  n_{i}n_{j}\delta_{\ast}\right)  =\left(  1/3\right)
\delta_{ij}\delta\left(  \mathbf{x}\right)  $ in $\mathbb{R}^{3},$ this yields
the well known formula of Frahm \cite{Frahm}%
\begin{equation}
\frac{\overline{\partial}^{2}}{\partial x_{i}\partial x_{j}}\left(  \frac
{1}{r}\right)  =\mathrm{p.v.}\left(  \frac{3x_{i}x_{j}-r^{2}\delta_{ij}}%
{r^{5}}\right)  -\left(  \frac{4\pi}{3}\right)  \delta_{ij}\delta\left(
\mathbf{x}\right)  \,. \label{PD.10}%
\end{equation}
We also immediately obtain that
\begin{equation}
\frac{\partial_{0}^{\ast2}\mathcal{P}f\left(  r^{-1}\right)  }{\partial
x_{i}\partial x_{j}}=\mathcal{P}f\left(  \frac{3x_{i}x_{j}-r^{2}\delta_{ij}%
}{r^{5}}\right)  +4\pi\left(  \delta_{ij}-4n_{i}n_{j}\right)  \delta_{\ast}\,,
\label{h.6}%
\end{equation}
a formula that can also be proved by other methods \cite{YE3}. On the other
hand, in \cite{Franklin} one can find the computation of%
\begin{equation}
\frac{\partial_{0}^{\ast}}{\partial x_{i}}\left(  \frac{\partial_{0}^{\ast}%
}{\partial x_{j}}\left(  \frac{1}{r}\right)  \right)  =\mathcal{P}f\left(
\frac{3x_{i}x_{j}-r^{2}\delta_{ij}}{r^{5}}\right)  -4\pi n_{i}n_{j}%
\delta_{\ast}\,. \label{h.7}%
\end{equation}
The fact that $\displaystyle\frac{\partial_{0}^{\ast}}{\partial x_{i}}\left(
\frac{\partial_{0}^{\ast}}{\partial x_{j}}\right)  \neq\frac{\partial
_{0}^{\ast2}}{\partial x_{i}\partial x_{j}}$ is obvious in the Example
\ref{Example h.1}, but it is harder to see it in cases like this
one\footnote{That the two results are different is overlooked in
\cite{Franklin}.}. Observe that the projection of both $4\pi\left(
\delta_{ij}-4n_{i}n_{j}\right)  \delta_{\ast}$ and of $-4\pi n_{i}n_{j}%
\delta_{\ast}$ onto $\mathcal{D}^{\prime}\left(  \mathbb{R}^{3}\right)  $\ is
given by $-\left(  4\pi/3\right)  \delta_{ij}\delta\left(  \mathbf{x}\right)
,$ but this does not mean that they are equal; observe also that one needs the
finite part in (\ref{h.6}) and in (\ref{h.7}) since the principal value, as
used in (\ref{PD.10}), exists in $\mathcal{D}^{\prime}\left(  \mathbb{R}%
^{3}\right)  $\ but not in $\left(  \mathcal{D}_{\ast}^{\left[  0\right]
}\left(  \mathbb{R}^{3}\right)  \right)  ^{\prime}.$\smallskip
\end{example}

.

\end{document}